\documentclass[a4paper,USenglish]{lipics-v2016}
 
\usepackage{booktabs}
\usepackage{xspace}
\usepackage{microtype}

\title{The Generalized Subterm Criterion in \protect\tttt\footnote{%
  This work was supported by FWF (Austrian Science Fund) project P27502.}}

\author{Christian Sternagel}
\affil{University of Innsbruck, Austria\\
  \texttt{christian.sternagel@uibk.ac.at}}
\authorrunning{C. Sternagel} 

\makeatletter
\renewcommand*\Copyright[1]{%
  \def\@Copyright{%
      \ifx#1\@empty \else \textcopyright\ #1;\\\fi
    }}
\def\copyrightline{%
  \ifx\@EventLogo\@empty
  \else
    \setbox\@tempboxa\hbox{\includegraphics[height=42\p@]{\@EventLogo}}%
    \rlap{\hspace\textwidth\hspace{-\wd\@tempboxa}\hspace{\z@}%
          \vtop to\z@{\vskip-0mm\unhbox\@tempboxa\vss}}%
  \fi
  \scriptsize
  \vtop{\hsize\textwidth
    \nobreakspace\\
    \@Copyright
    \ifx\@EventLongTitle\@empty\else\@EventLongTitle.\\\fi
    \ifx\@EventEditors\@empty\else
      \@Eds: \@EventEditors
      ; Article~No.\,\@ArticleNo; pp.\,\@ArticleNo:\thepage--\@ArticleNo:\pageref{LastPage}.%
    \fi
  }
}
\makeatother
\Copyright{Christian Sternagel}

\subjclass{%
  F.4.2 Grammars and Other Rewriting Systems 
}
\keywords{%
  termination,
  subterm criterion,
  SMT encodings}

\EventEditors{Aart Middeldorp and René Thiemann}
\EventNoEds{2}
\EventLongTitle{Proceedings of the 15th International Workshop on Termination}
\EventShortTitle{WST 2016}
\EventLogo{}
\ArticleNo{11}
\DOIPrefix{}

\newcommand\isafor{\textsf{Isa\kern-0.2exF\kern-0.2exo\kern-0.2exR}\xspace}
\newcommand\ceta{\textsf{C\kern-0.2exe\kern-0.5exT\kern-0.5exA}\xspace}
\newcommand\tttt{%
  \textsf{T\kern-0.2em\raisebox{-0.5ex}T\kern-0.2emT\kern-0.2em\raisebox{-0.5ex}2}\xspace}

\DeclareMathAlphabet{\mysf}{\encodingdefault}{\sfdefault}{m}{n}

\newcommand\var[1]{\textsf{\upshape#1}}
\newcommand\form[1]{\textsf{\upshape\uppercase{#1}}}
\newcommand\pos[2]{\var{p}^{#2}_{#1}}
\newcommand\wt[2]{\var{w}^{#2}_{#1}}
\newcommand\MULT{\form{M}}
\newcommand\GEQ{\form{GEQ}}
\newcommand\UPPER{\form{UPPER}}
\newcommand\NEQ{\form{NEQ}}
\newcommand\RT{\form{RT}}
\newcommand\SAN{\form{SAN}}

\newcommand\mulex[1]{#1_{\mathsf{mul}}}
\newcommand\rel{\succ}
\newcommand\relinv{\prec}
\newcommand\restr[2][\rel]{\mathrel{#1_{\downarrow#2}}}
\newcommand\subt{\lhd}

\newcommand\supt{\rhd}
\newcommand\supteq{\unrhd}
\newcommand\upper{\mathrm{upper}}
\newcommand\encmul{\form{MUL}}
\newcommand\THEN{\mathrel{\textsf?}}
\newcommand\ELSE{\mathrel{\textsf:}}
\newcommand\ITE[3]{#1\THEN#2\ELSE#3}
\newcommand\SUM{\sum}
\renewcommand\implies{\longrightarrow}
\newcommand\funs{\mathcal{F}}
\newcommand\mults{\mathcal{M}}
\newcommand\nats{\mathbb{N}}
\newcommand\etal{et\ al{.}}
\newcommand\suptmulex[1][\pi]{\mulex{\supt}^{#1}}
\newcommand\suptmulexeq[1][\pi]{\mulex{\supteq}^{#1}}
\newcommand\trs[1]{\mathcal{#1}}
\let\oldeqref\eqref
\renewcommand\eqref[1]{\oldeqref{eq:#1}}
\newcommand\thmref[1]{Theorem~\ref{thm:#1}}
\newcommand\lemref[1]{Lemma~\ref{lem:#1}}

\newcommand\Subterms{\mathcal{S}\mathsf{ub}}

\newcommand\PR{(\trs{P},\trs{R})}

\begin{document}

\maketitle

\begin{abstract}
We present an SMT encoding of a generalized version of the subterm criterion and
evaluate its implementation in \tttt.
\end{abstract}

\section{Preliminaries}

We assume basic familiarity with term rewriting \cite{BN98} in
general and the dependency pair framework \cite{GTS05} for proving termination
in particular. We start with a recap of terminology and notation that we use in
the remainder.

By $\mults(A)$, we denote the set of \emph{finite multisets} ranging over
elements from the set $A$. We write $M(x)$ for the \emph{multiplicity} (i.e.,
number of occurrences) of $x$ in the multiset $M$, use $+$ for multiset sum,
but otherwise use standard set-notation.

Given a relation $\rel$, its \emph{restriction to the set $A$}, written
$\restr[\rel]{A}$, is the relation defined by the set
$\{(x, y) \mid x \rel y, x \in A, y \in A\}$.
Moreover, for any function $f$, we use $x \rel^f y$ as a shorthand for
$f(x) \rel f(x)$.

The \emph{multiset extension} $\mulex{\rel}$ of a given relation $\rel$ is
defined by:
\begin{equation*}
M \mulex{\rel} N \text{ iff } \exists X\;Y\,Z.\:
X \neq \varnothing, M = X + Z, N = Y + Z, \forall y \in Y.\: \exists x \in X.\: x \rel y
\end{equation*}
A useful fact about the multiset extension is that we may always ``maximize''
the common part $Z$ in the above definition.
\begin{lemma}\label{lem:disj}
Consider an irreflexive and transitive relation $\rel$ and multisets $M$, $N$
such that $M \mulex{\rel} N$. Moreover, let $X = M - M \cap N$ and $Y = N - M
\cap N$. Then $X \neq \varnothing$ and $\forall y \in Y.\: \exists x \in X.\: x
\rel  y$.
\end{lemma}
While intuitively obvious, a rigorous proof of this fact does not seem to be
widely known.\footnote{An alternative proof of this fact is indicated in Vincent
van Oostrom's PhD thesis \cite{O94}.}
In preparation for the proof, we recall the following easy fact about finite
relations.
\begin{lemma}\label{lem:finite-wf}
Every finite, irreflexive, and transitive relation is well-founded.
\end{lemma}
\begin{proof}
Let $\rel$ be a finite, irreflexive, and transitive relation. For the sake of a
contradiction, assume that $\rel$ is not well-founded. Then there is an infinite
sequence
$
a_1 \rel a_2 \rel a_3 \rel \cdots
$
whose elements are in the finite (since $\rel$ is finite) field of $\rel$. But
then, by the (infinite) pigeonhole principle, there is some recurring element
$a_i$, i.e., $\cdots \rel a_i \rel \cdots \rel a_i \rel \cdots$. By
transitivity we obtain $a_i \rel a_i$ contradicting the irreflexivity of $\rel$.
\end{proof}
Noting that the converse of any finite, irreflexive, and transitive relation is
again finite, irreflexive, and transitive, \lemref{finite-wf} allows us to
employ well-founded induction where the induction hypothesis holds for
``bigger'' elements, as exemplified in the following proof.
\begin{proof}[Proof of \lemref{disj}]
Since $M \mulex{\rel} N$ we obtain $I \neq \varnothing$, $J$, and $K$ such that
$M = I + K$, $N = J + K$, and $\forall j \in J.\: \exists i \in I.\: i \rel j$.
Let $A = I - I \cap J$, $B = J - I \cap J$, and consider the finite set $D$ of
elements occurring in either of $I$ and $J$. Now, appealing to
\lemref{finite-wf}, we employ well-founded induction
with respect to $\restr[\relinv]{D}$ in order to prove:
\begin{equation}
\forall j \in J.\: \exists a \in A.\: a \rel j\tag{$\dagger$}\label{eq:disj}
\end{equation}
Thus we assume $j \in J$ for some arbitrary but fixed $j$ and obtain the
induction hypothesis (IH)
\mbox{$\forall c \restr{D} j.\: c \in J \implies \exists a \in A.\: a \rel c$}.
From $j \in J$ we obtain an $i \in I$ with $i \rel j$. Now if $i \in A$, then we
are done. Otherwise, $i \in J$ and by IH we obtain an $a \in A$ with $a \rel i$.
Since $\rel$ is transitive, this implies $a \rel j$, concluding the proof of
\eqref{disj}. But then also $\forall b \in B.\: \exists a \in A.\: x \rel b$
and $A \neq \varnothing$.
We conclude by noting the following two equalities:
\begin{alignat*}{5}
X &= M - M \cap N &&= (I + K) - (I + K) \cap (J + K) &&= I - I \cap J &&= A,\\
Y &= N - M \cap N &&= (J + K) - (I + K) \cap (J + K) &&= J - I \cap J &&= B.
\tag*{\qedhere}
\end{alignat*}
\end{proof}

\section{A Generalized Subterm Criterion}

Recall the subterm criterion -- originally by Hirokawa and
Middeldorp~\cite{HM04} and later reformulated as a processor for the dependency
pair framework -- which is a particularly elegant technique (due to its
simplicity and the fact that the $\trs{R}$-component of a dependency pair
problem $\PR$ may be ignored).
\begin{definition}[Simple projections]
A \emph{simple projection} is a function $\pi : \funs \to \nats$ that maps every
$n$-ary function symbol $f$ to some natural number $\pi(f) \in \{1,\ldots,n\}$.
Applying a simple projection to a term is defined by
$\pi(f(t_1,\ldots,t_n)) = t_{\pi(f)}$.
\end{definition}
\begin{theorem}
If $\trs{P} \subseteq {\supteq^\pi}$ for simple projection $\pi$,
then $\PR$ is finite iff $(\trs{P}\setminus{\supt^\pi},\trs{R})$ is.
\qed
\end{theorem}
Recall that the appropriate notion of \emph{finiteness} for the subterm
criterion is ``the absence of minimal infinite chains.''

For an AC-variant of the subterm criterion (i.e., a variant for rewriting modulo
associative and/or commutative function symbols), Yamada \etal~\cite{YSTK16}
generalized simple projections to so-called multiprojections.
\begin{definition}[Multiprojections]
A \emph{multiprojection} is a function $\pi : \funs \to \mults(\nats)$ that
maps every $n$-ary function symbol $f$ to a multiset $\pi(f) \subseteq
\mults(\{1,\ldots,n\})$.  Applying a multiprojection to a term yields a multiset of terms as
follows:
\[
\pi(t) = \begin{cases}
\pi(t_{i_1}) + \cdots + \pi(t_{i_k})
  & \text{if $t = f(t_1,\ldots,t_n)$ and $\pi(f) = \{i_1,\ldots,i_k\} \neq \varnothing$,}\\
\{t\} & \text{otherwise.}
\end{cases}
\]
We write $s \suptmulexeq t$ if either $s \suptmulex t$ or $\pi(s) = \pi(t)$.
\end{definition}
A compromise between simple projections and full multiprojections is to allow
recursive projections (possibly through defined symbols). While theoretically
subsumed by multiprojections, we included such recursive projections in our
experiments in order to assess their performance in practice.

The following is a specialization of the AC subterm criterion by Yamada
\etal~\cite[Theorem~33]{YSTK16} to the non-AC case.
\begin{theorem}
\label{thm:gsc}
Let $\pi$ be a multiprojection such that $\trs{P} \subseteq {\suptmulexeq}$ and
$f(\ldots) \suptmulexeq r$ for all $f(\ldots) \to r \in \trs{R}$ with $\pi(f) \neq
\varnothing$. Then $\PR$ is finite iff
$(\trs{P}\setminus{\suptmulex}, \trs{R})$ is.
\qed
\end{theorem}
This result (which is also formalized in \isafor~\cite{TS09}) states the soundness of
a generalized version of the subterm criterion and thus gives the theoretical
backing for implementing such a technique in a termination tool. In the
following we are concerned with the more practical problem of an efficient
implementation.

That is, given a DP problem~$\PR$ we want to find a
multiprojection $\pi$ that satisfies the conditions of \thmref{gsc} and orients
at least one rule of $\trs{P}$ strictly by $\suptmulex$.

Since the problem of finding such a multiprojection seems similar to the problem
of finding an appropriate argument filter for a reduction pair \cite{CSLTG06}, and
the latter has been successfully tackled by various kinds of SAT and SMT
encodings, we take a similar approach.

\section{Implementation and Experiments}

There are basically two issues that have to be considered: (1)
how to encode a multiprojection $\pi$ and thereby the multiset $\pi(s)$, and (2)
how to encode the comparison between two encodings of multisets with respect to
the multiset extension of $\supt$.

In the following we use \var{lowercase sans serif} for propositional and
arithmetical variables, and \form{uppercase sans serif} for functions that
result in formulas.

\subparagraph*{Encoding Multiprojections.}

We encode the multiplicity of a term~$t$ in the multiset $\pi(s)$,
which is $0$ if $t$ does not occur in $\pi(s)$ at all,
by $\MULT_s(t) =
\encmul(1,s,t)$. The latter is defined as follows
\[
\encmul(w, s, t) = \begin{cases}
\ITE{\left(\displaystyle\bigwedge_{1\leq i\leq n} \lnot \pos{f}{i} \right)}{w}{0}
  & \text{if $s = t = f(t_1,\ldots,t_n)$}\\
w & \text{if $s = t$ and $t$ is a variable}\\
\displaystyle\SUM_{1\leq i\leq n}(\ITE{\pos{f}{i}}{\encmul(w\cdot\wt{f}{i},s_i,t)}{0})
  & \text{if $t \subt s = f(s_1,\ldots,s_n)$}\\
0 & \text{otherwise}
\end{cases}
\]
where $\ITE{b}{t}{e}$ denotes \emph{if $b$ then $t$ else $e$} and the intended
meaning of variables is that $\pos{f}{i} = \top$ precisely when \emph{$\pi$
projects to the $i$-th argument of $f$}, in which case $\wt{f}{i}$ gives the
\emph{weight of $i$ in $\pi(f)$}, i.e., its number of occurrences in
$\pi(f)$.\footnote{In experiments, replacing $\pos{f}{i} = \top$ by $\wt{f}{i} >
0$ resulted in a slightly increased number of timeouts.}

\subparagraph*{Encoding Multiset Comparison.}

Now consider the problem of finding $\pi$ such that $s \suptmulex t$ for given
terms $s$ and $t$. Noting that, independent of the exact $\pi$, $\pi(s)$ and
$\pi(t)$ are multisets over the finite set of subterms of $s$
and $t$, it suffices to find an encoding for comparing multisets over finite
domains. This allows us to make use of the following observation.
\begin{lemma}[Comparing multisets over finite domains]
\label{lem:finite-mulex}
Let $D$ be a finite set, and $M, N \subseteq \mults(D)$.
Then, for irreflexive and transitive $\rel$, $M \mulex{\rel} N$ is equivalent to
\begin{equation}
\forall d\in D.\: \upper(d) \implies M(d) \geq N(d)
\text{ and }
M \neq N\tag{$\star$}\label{eq:mulex}
\end{equation}
where $\upper(x)$ iff
$
\forall d \in D.\: d \rel x \implies M(d) = N(d)
$.
\end{lemma}
\begin{proof}
We start with the direction from \eqref{mulex} to $M \mulex{\rel} N$.  Assume
\eqref{mulex} for $M$ and $N$, and define the multisets
$Z = \{x \in M \cap N \mid \upper(x)\}$,
$X = M - Z$, and
$Y = N - Z$ (i.e., $M = X + Z$ and $N = Y + Z$).
Then, appealing to \lemref{finite-wf}, we use well-founded induction with
respect to $\restr[\relinv]{D}$ in order to prove
\begin{equation}
\forall y \in Y.\: \exists x \in X.\: x \rel y\tag{$\ddagger$}\label{eq:yless}
\end{equation}
Thus we assume $y \in Y$ for some arbitrary but fixed $y$ and obtain the
induction hypothesis (IH) $\forall z
\restr{D} y.\: z \in Y \implies \exists x \in X.\: x \rel z$.
Also note that $\lnot \upper(y)$, since otherwise $M(y)
\geq N(y)$ by \eqref{mulex} and thus $Z(y) = N(y)$, contradicting $y \in Y$.
Therefore, we obtain $z \rel y$ with $M(z) \neq N(z)$ by definition of $\upper$.
Now, either $M(z) > N(z)$ or $N(z) > M(z)$. In the former case $z \in X$ and we
are done. In the latter case $z \in Y$ and thus we obtain an $x \in X$ such that
$x \rel z$ by IH and conclude \eqref{yless} by transitivity of $\rel$. It
remains to show $X \neq \varnothing$. Since $M \neq N$ there is some $x$ with
$M(x) \neq N(x)$. If $M(x) > N(x)$, then $x \in X$ and we are done.  Otherwise,
$N(x) > M(x)$ and thus $x \in Y$ and we conclude by invoking \eqref{yless}.

For the other direction, assume $M \mulex{\rel} N$. Then for
$Z = M \cap N$, $X = M - Z$, and $Y = M - Z$,
we have $X \neq \varnothing$, $X \cap Y = \varnothing$,
$M = X + Z$, $N = Y + Z$ and $\forall y \in Y.\: \exists x \in X.\: x \rel y$,
using \lemref{disj}. This further implies $M \neq N$.
Now assume $d \in D$ and $\upper(d)$.  Then either $d \in Y$ or $d \notin Y$. In
the latter case, clearly $M(d) \geq N(d)$, and we are done. In the former case,
we obtain an $x \in X$ with $x \rel d$. Moreover, since $X \cap Y =
\varnothing$, we have $x \notin Y$.  But then $M(x) \neq N(x)$, contradicting
$\upper(d)$.
\end{proof}

\subparagraph*{Encoding the Generalized Subterm Criterion.}
Putting everything together we obtain the encoding
\begin{gather*}
(\forall s \to t \in \trs{P}.\: \GEQ(s, t)) \land
(\exists s \to t \in \trs{P}.\: \NEQ(s, t)) \land {}\\
(\forall s \to t \in \trs{R}.\: \RT(s) \implies \GEQ(s, t)) \land
(\forall f \in \funs\PR.\: \SAN(f))
\end{gather*}
where
\begin{align*}
\GEQ(s, t) &\text{ iff }
\forall u \in \Subterms(s, t).\: \UPPER(u) \implies \MULT_s(u) \geq \MULT_t(u)
\\
\UPPER(u) &\text{ iff } \forall v \in \Subterms(s, t).\:
v \supt u \implies \MULT_s(v) = \MULT_t(v)
\\
\NEQ(s, t) &\text{ iff }
\lnot(\forall u \in \Subterms(s, t).\: \MULT_s(u) = \MULT_t(u))
\\
\RT(f(s_1,\ldots,s_n)) &\text{ iff }
\exists 1\leq i \leq n.\: \pos{f}{i}.
\\
\SAN(f) &\text{ iff } \bigwedge_{1\leq i \leq \textsf{arity}(f)}
\left(\pos{f}{i} \implies \wt{f}{i} > 0\right)
\end{align*}
Here $\Subterms(s, t)$ denotes the set of all (i.e., including $s$ and $t$
themselves) subterms of $s$ and $t$, and
$\SAN$ is a ``sanity check'' that makes sure that propositional and
arithmetical variables play well together.
Every satisfying assignment gives rise to a multiprojection $\pi$ satisfying the
conditions of \thmref{gsc}.

\subparagraph*{Experiments.}

\begin{table}
\caption{\label{tab:experiments}Experiments on 1498 standard TRSs of TPDB 10.3}
\def\twoc#1{\multicolumn{2}{c}{#1}}%
\def\C#1{\multicolumn{1}{c}{#1}}
\centering
\begin{tabular}{lccccccc}
\toprule
& \twoc{Yes} & \twoc{Maybe} & \twoc{Timeout} \\
\cmidrule(lr){2-3}
\cmidrule(lr){4-5}
\cmidrule(lr){6-7}
Projections & \C\# & \C{(sec)} & \C\# & \C{(sec)} & \C\# & \C{(sec)} & Total (sec) \\
\midrule
simple    & 265 & 31.1 & 1184 & 226.8 & 49 & 254.0 & 502.9\\
recursive & 292 & 35.4 & 1155 & 240.4 & 51 & 255.0 & 530.9\\
multi     & 351 & 61.2 & 1081 & 419.0 & 66 & 330.0 & 810.2\\
\addlinespace
all       & 352 & 30.4 & 1099 & 230.3 & 47 & 235.0 & 495.7\\
\bottomrule
\end{tabular}
\end{table}

We conducted experiments in order to assess our implementation. To this end we
took all the 1498 TRSs in the standard (as in ``standard term rewriting'')
category of the termination problem database (TPDB) version 10.3 and tried to
prove their termination with the following strategy: first compute dependency
pairs, then compute the estimated dependency graph~$\mathcal{G}$, and finally
try repeatedly to either decompose $\mathcal{G}$ into strongly connected
components or apply the subterm criterion. For the subterm criterion we tried
either \emph{simple projections} (simple), \emph{recursive projections}
(recursive), \emph{multiprojection} (multi), or a parallel combination of those
(all).

In summary, the parallel combination of different kinds of projections results
in a significant increase of power (i.e., number of yeses) and does not have a
negative impact on the speed, compared to the original implementation of the
subterm criterion (simple) of \tttt~\cite{KSZM09}.

Encouraged by this results, we incorporated our new implementation also into the
competition strategy of \tttt and compared it to its 2015 competition version.
In this way, we were able to obtain 12 additional yeses. However, each of those
12 systems could already be handled by some other termination tool in the 2015
termination competition.

\subparagraph*{Acknowledgments.}

We thank Vincent van Oostrom for pointing us to \lemref{finite-mulex} and
Bertram Felgenhauer for helpful discussion concerning $\encmul$.
We further thank the Austrian Science Fund (FWF project P27502) for supporting
this work.

\appendix


\bibliography{short,references}

\end{document}